\documentclass[conference, 10 pt, twocolumn]{ieeeconf}
\IEEEoverridecommandlockouts
\usepackage{bbm}
\usepackage{enumerate}
\usepackage{times}
\usepackage{textcomp}
\usepackage{cite}
\usepackage{url}
\usepackage{subcaption}
\usepackage{amsfonts,mathrsfs}
\usepackage{amssymb,amsmath}
\usepackage{verbatim}
\usepackage{acronym}
\usepackage{mathtools}
\def\argmin{\text{argmin}}
\newcommand{\beh}[1]{{\color{black} #1}}

\usepackage{graphicx}
\usepackage{todonotes}
\usepackage{hyperref}
\usepackage{amsmath,amssymb}
\usepackage{algorithm}
\usepackage{algorithmic}
\usepackage{romannum}
\usepackage{graphicx}
\usepackage{subcaption}
\usepackage{ulem}
\usepackage{xspace}

\DeclareMathOperator{\trace}{Trace}
\DeclareMathOperator{\diag}{diag}

\newcommand{\x}{\boldsymbol{x}}

\newcommand{\e}{\mathbf{e}}
\newcommand{\bu}{\boldsymbol{u}}
\newcommand{\bv}{\boldsymbol{v}}
\newcommand{\allzero}{\boldsymbol{0}}
\newcommand{\allone}{\mathbf{1}}
\newcommand{\R}{\mathbb{R}}
\newcommand{\G}{\mathcal{G}}
\newcommand{\Edge}{\mathcal{E}}
\newcommand{\A}{\mathcal{A}}
\newcommand{\Oc}{\mathcal{O}}
\newcommand{\range}{\mathcal{R}}
\newcommand{\bseta}{\boldsymbol{\eta}}

\newcommand{\s}{\boldsymbol{s}}

\newtheorem{theorem}{Theorem}

\newtheorem{definition}{Definition}

\newtheorem{proposition}[theorem]{Proposition}
\newtheorem{remark}{Remark}

\newcommand{\ucenter}{{U-centrality}\xspace}

\newcommand{\xin}[1]{{\color{black} #1}}
\newcommand{\Xin}[1]{{\color{black} #1}}

\begin{document}
\title{\LARGE \bf \ucenter: A Network Centrality Measure Based on Minimum Energy Control for Laplacian Dynamics}

\author{Xinran Zheng, Leonardo Massai, Massimo Franceschetti, and Behrouz Touri
\thanks{This research was partially  supported by AFOSR FA9550-24-1-0129 and AFOSR FA9550-23-1-0057 grants.}
\thanks{X. Zheng and B. Touri are with the Department of Industrial and Enterprise Systems Engineering, University of Illinois Urbana-Champaign.
        {\tt\small \{xinran22, touri1\}@illinois.edu}}%
\thanks{L. Massai is with the Automatic Control Laboratory and the NCCR Automation, EPFL.
        {\tt\small l.massai@epfl.ch}}%
\thanks{M. Franceschetti is with the Electrical and Computer Engineering Department, University of California San Diego.
        {\tt\small mfranceschetti@ucsd.edu}}%
}

\maketitle

\begin{abstract}
Network centrality is a foundational concept for quantifying the importance of nodes within a network. Many traditional centrality measures—such as degree and betweenness centrality—are purely structural and often overlook the dynamics that unfold across the network. However, the notion of a node’s importance is inherently context-dependent and must reflect both the system’s dynamics and the specific objectives guiding its operation. Motivated by this perspective, we propose a dynamic, task-aware centrality framework rooted in optimal control theory. \Xin{By formulating a problem on minimum energy control of average opinion based on Laplacian dynamics and focusing on the variance of terminal state, we introduce a novel centrality measure—termed \ucenter—that quantifies a node’s ability to unify the agents' state.} We demonstrate that \ucenter interpolates between known measures: it aligns with degree centrality in the short-time horizon and converges to a new centrality over longer time scales which is closely related to current-flow closeness centrality. This work bridges structural and dynamical approaches to centrality, offering a principled, versatile tool for network analysis in dynamic environments.
\end{abstract}


\thispagestyle{empty}
\pagestyle{empty}

\section{Introduction}

Node centrality and ranking are key concepts in complex network analysis that measure and rank the importance of nodes in a network. In terms of practical applications, they are widely used to identify key infrastructural nodes in complex engineering networks (such as the Internet \cite{iyer2013attack} and electrical networks \cite{hines2008centrality}), brain networks \cite{fraschini2014eeg,joyce2010new}, economic networks \cite{ballester2006s,banerjee2013diffusion}, and epidemic networks \cite{rodrigues2019network}.

Several different centrality measures have been proposed, such as degree centrality, closeness centrality, betweenness centrality, eigenvector centrality, and PageRank centrality \cite{newman2018networks}. It is noticeable that all of these centrality measures are ``intrinsic'',  namely their definitions depend only on the network structure. However, since the importance of nodes depends on context and application, there is not a unique definition of network centrality, and there is no objective viewpoint on how one should choose between the plethora of available centrality notions in any given application setting. In an effort to achieve such a viewpoint, we argue that centrality can be better understood in the context of a dynamical process over a network and of an objective function relating to controllability of the dynamics. Accordingly, we propose a novel approach to network centrality based on a dynamical system perspective and tie centrality measures to network dynamics and control objectives.

\xin{

More specifically, we consider the following problem: given a dynamic model over a network, determine which node is the best one to control in order to reach a certain control objective, or in other words, controlling which node would achieve the best system objective. Controllability of complex networks and multi-agent systems concerning the underlying graph has been extensively studied in \cite{chapman2014controllability,o2016conjecture,tanner2004controllability,rahmani2009controllability,aguilar2014graph,pasqualetti2014controllability,zhao2019networks}, to name just a few. Among them, \cite{chapman2014controllability} and \cite{o2016conjecture} study the controllability of composite networks and random networks, respectively, while \cite{tanner2004controllability,rahmani2009controllability,aguilar2014graph} study the controllability of Laplacian leader-follower dynamics. Rather than viewing controllability as a possibility concept, the authors in \cite{pasqualetti2014controllability} introduce a controllability metric, where they quantify the difficulty of controlling a dynamical system by the minimum control energy required to transfer the state from the origin to the worst point on the unit sphere. They consider general linear time-invariant (LTI) dynamics, and provide a bound on the energy in terms of the number of controlled nodes and some heuristic rules on how to select the controlled nodes. 

Motivated by the minimum energy control problem in \cite{pasqualetti2014controllability} and the Laplacian leader-follower dynamics in \cite{tanner2004controllability,rahmani2009controllability,aguilar2014graph}, we study a Laplacian dynamics problem where we aim to find the minimum energy to steer the state from $\allzero$ to a state with aggregate state exceeding a certain threshold. Interestingly, we find that by controlling any one of the nodes, the minimum energy required remains the same. Given that, in our setting, we define the best node as the node that can be controlled with minimum energy to achieve the minimum distance between the terminal state and the consensus state. \beh{In other words, the central node is the leader who is the most capable of unifying the agents while the aggregate state passes a certain threshold. As a result, we refer to it as \ucenter.} In addition, 
we derive the exact expression (instead of a bound) of the energy and the distance of interest for each node based on the topology of the graph, and regard the distance as a measure of node centrality.
}

\noindent\textbf{Contributions}.
First, we formulate a minimum energy control problem for Laplacian dynamics over a network, where the objective is to steer the state vector from the origin to a subspace where the aggregate state of agents is greater than a given value within time $t_f$. Next, we define a new centrality measure ``\ucenter'' representing the $l_2$ distance between the consensus state and the terminal state given by the minimum energy control solution when only the respective node is controlled. 
Then, we show that for short-term and long-term influence, \ucenter would be closely related to existing centrality measures in network science. More specifically, when $t_f\approx0$, \ucenter coincides with degree centrality. When $t_f\gg0$, \ucenter provides a new centrality measure that is closely related to current-flow closeness centrality~\cite{bozzo2013resistance}. Through the study of \ucenter for $t_f\gg 0$ in trees, we show how it is informative of the position of the nodes in the graph, particularly in relation to node peripherality in such graphs.

\paragraph*{Notations}
We denote the set of real numbers by $\R$ and the vector space of $n$-dimensional real-valued column vectors by $\R^n$. We use bold lower-case letters to denote column vectors. We use $I$ to denote the identity matrix of a known underlying dimension $n$. We use $\allone$, $\allzero$, and $\e_i$ to denote the all-one vector, the all-zero vector, and the vector with 1 in the $i$th coordinate and 0's elsewhere in $\R^n$, respectively. We use $[n]$ to denote $\{1, \ldots, n\}$. We use $\| \cdot \|_2$ to denote the standard Euclidean norm. 
We use $\range(A)$ to denote the column range of matrix $A$.
For an undirected graph $\G=([n],\Edge)$ with $n$ vertices and a set of edges $\Edge$, define $\A$ to be the adjacency matrix of $\G$ with $\A_{ij}=1$ if $(i,j)\in\Edge$ and $\A_{ij}=0$, otherwise.
\xin{We say $h(t)=\Oc(g(t))$ if there exists a positive real number $M$ and a real number $t_0$ such that ${|h(t)|\leq M|g(t)|}$ for all $t\geq t_0$.}

\section{Problem Formulation}
In this section, we formulate an optimal control problem and define a new centrality measure in this setting. First, let us introduce some preliminary concepts on minimum energy control and Laplacian dynamics.

\subsection{Preliminary on Minimum Energy Control}
Consider a dynamical model constituted of a network of $n$ agents whose interactions can be modeled by a linear time-invariant (LTI) dynamics and a minimum energy control problem
\begin{align}
\label{eqn:min_energy_control}
    \min_{\bu(t)} &\int_{0}^{t_f} \|\bu(t)\|_2^2 \,dt \\
    \text{Subject to: } &\dot{\x}(t)=A\x(t)+B\bu(t) \cr
    &\x(0)=\allzero \cr
    &\x(t_f)\in U, \nonumber
\end{align}
where $\x(t)=\begin{bmatrix} x_1(t)&\cdots&x_n(t) \end{bmatrix}^\intercal\in\R^n$ is the time-varying state vector, and $\bu(t)\in\R^p$ is the time-varying external control input. 
$A\in\R^{n\times n}$ models the influence between the agents and conforms with an underlying graph $\G=([n],\Edge)$, i.e., $A_{ij}>0$ iff $(i,j)\in \Edge$ for $i\not=j$. 
In this setting, \xin{$B\in\R^{n\times p}$ determines the nodes that are influenced by the control inputs. More specifically, if a subset of nodes $\{k_1,\ldots,k_p\}\subseteq[n]$ are controlled, then we assume that we have $p$ independent control inputs that are injected through these $p$ nodes, i.e., $B=\begin{bmatrix} \e_{k_1}&\cdots&\e_{k_p} \end{bmatrix}$ and we are interested in the minimum energy control of such a network through these  nodes. }
In the above problem, $U\subseteq \R^n$ is the terminal set of interest that has a non-empty intersection with the reachable subspace of this LTI system.

For $U=\{\x_f\}$, where $\x_f$ is a reachable state, the solution to \eqref{eqn:min_energy_control} is given by (see Theorem 11.4 in \cite{hespanha2018linear})
\begin{align*}
    \bu^*(t)=B^\intercal e^{A^\intercal(t_f-t)}\bseta_f,
\end{align*}
where $\bseta_f\in\R^n$ can be any vector with $\x_f=W_R\bseta_f$, and $W_R$ is the reachability Gramian
\begin{align}\label{eqn:WR}
    W_R=\int_{0}^{t_f} e^{A\tau}BB^\intercal e^{A^\intercal\tau} \,d\tau.
\end{align}
In this case, the corresponding minimum energy is
\begin{align}
\label{eqn:min_energy_general}
    E=\bseta_f^\intercal W_R \bseta_f.
\end{align}


\subsection{Preliminary on Laplacian Dynamics}

Given an undirected graph $\G$ with adjacency matrix $\A$, we define the Laplacian matrix of graph $\G$ to be ${L=\diag(\boldsymbol{d})-\A}$, where the degree vector $\boldsymbol{d}=\A\allone$ and $\diag(\boldsymbol{d})$ is the diagonal matrix with $i$th diagonal element equal to $d_i$. Then, if we let $A=-L$ and $B=0$ in \eqref{eqn:min_energy_control}, we get the Laplacian dynamics $\dot{\x}=-L\x$, which is the continuous version of the French-Degroot opinion dynamics model~\cite{french1956formal,degroot1974reaching}. It is easy to verify that for the dynamics $\dot{\x}=-L\x$, we have $\allone^\intercal\x(t)=\allone^\intercal\x(0)$ for all $t$, i.e., the average opinion of all agents in the network does not change. In addition, we have $\lim_{t\rightarrow\infty}\x(t)=\frac{\allone^\intercal\x(0)}{n}\allone$, i.e., the opinion dynamics converges to the consensus state.

\subsection{Minimum Energy Control of Average Opinion}
For the Laplacian dynamics, consider the optimal control problem~\eqref{eqn:min_energy_control} where the goal is to drive the state from $\allzero$ to \beh{the terminal set $U=\{\allone^\intercal \x\geq c\}$ for some $c>0$, where the aggregate state surpasses a certain threshold $c$. In other words, consider} 
\begin{align}
\label{eqn:ave_state_control}
    \min_{\bu(t)} &\int_{0}^{t_f} \|\bu(t)\|_2^2 \,dt \\
    \text{Subject to: } &\dot{\x}(t)=-L\x(t)+B\bu(t) \cr
    &\x(0)=\allzero \cr
    &\allone^\intercal\x(t_f):=\allone^\intercal\x_f\geq c. \nonumber
\end{align}
We state the solution to \eqref{eqn:ave_state_control} in the following proposition. \begin{proposition}
    The minimum control energy of \eqref{eqn:ave_state_control} is given by
    \begin{align*}
        E=\frac{c^2}{\allone^\intercal W_R\allone},
    \end{align*}
    and the corresponding terminal state is $ \x_f=\frac{c}{\allone^\intercal W_R\allone}W_R\allone$,
    where $W_R$ is the the reachability Gramian given in~\eqref{eqn:WR}.  
\end{proposition}

\begin{proof}
    Since the minimum energy required to steer the state form $\allzero$ to $\x_f\in \range(W_R)$ is given by \eqref{eqn:min_energy_general}, solving \eqref{eqn:ave_state_control} is equivalent to solving 
    \begin{align}
    \label{eqn:ave_state_control_uc}
        \min_{\bseta_f}\quad &\bseta_f^\intercal W_R\bseta_f \\
        \text{Subject to: } &\allone^\intercal W_R\bseta_f\geq c.\nonumber
    \end{align}
    \beh{Note that the above problem is a convex (quadratic) problem as the controllability Gramian $W_R$ in~\eqref{eqn:WR} is positive semidefinite and we have a linear  constraint set. }
    To solve \eqref{eqn:ave_state_control_uc}, consider the Lagrangian of this quadratic program 
    \begin{align*}
        \mathcal{L}(\bseta_f,\mu)=\bseta_f^\intercal W_R\bseta_f+\mu(-\allone^\intercal W_R\bseta_f+c).
    \end{align*}
    Since the Slater condition holds for~\eqref{eqn:ave_state_control_uc},  $\bseta_f,\mu$ are optimal primal and dual variables, iff they satisfy the KKT conditions (see Chapter 5 in \cite{boyd2004convex}) 
        \begin{subequations}   
            \begin{align}\label{eqn:KKT_uc1}
                    2W_R\bseta_f-\mu W_R\allone&=0,\\\label{eqn:KKT_uc2}
                    \allone^\intercal W_R\bseta_f&\geq c,\\\label{eqn:KKT_uc3}
                    \mu&\geq0,\\\label{eqn:KKT_uc4}
                    \mu(-\allone^\intercal W_R\bseta_f+c)&=0.
                \end{align}
        \end{subequations}
    \beh{Note that to satisfy~\eqref{eqn:KKT_uc4}, either $\mu=0$ or $\allone^\intercal W_R\bseta_f=c$. But we cannot have $\mu=0$ as with that, ~\eqref{eqn:KKT_uc1} would imply that $\allone^\intercal W_R\bseta_f=0$, which contradicts~\eqref{eqn:KKT_uc2}. Therefore, $\mu\not=0$ and hence, $\allone^\intercal W_R\bseta_f=c$. Using this in~\eqref{eqn:KKT_uc1} we obtain the unique solution for the dual variable $\mu=\frac{2c}{\allone^\intercal W_R\allone}$. Finally, using this in~\eqref{eqn:KKT_uc1}, we arrive at the unique terminal state and the minimum energy for our optimal control problem}
    \begin{align*}
        \x_f&=W_R\bseta_f=\frac{c}{\allone^\intercal W_R\allone}W_R\allone,\\
        E&=\bseta_f^\intercal W_R\bseta_f=\frac{c^2}{\allone^\intercal W_R\allone}.
    \end{align*}
\end{proof}

Note that for such a Laplacian dynamics, we have $L^\intercal\allone=L\allone=\allzero$. 
Hence, $\allone^\intercal e^{-L\tau}=\allone^\intercal$ and $e^{-L\tau}\allone=\allone$. Therefore, we have
\begin{align*}
    \allone^\intercal W_R\allone=\int_{0}^{t_f} \allone^\intercal BB^\intercal\allone \,d\tau
    =t_f \allone^\intercal BB^\intercal\allone,
\end{align*}
and
    $W_R\allone=\int_{0}^{t_f} e^{-L\tau}BB^\intercal \allone \,d\tau$.
Next, let us look at the specific expressions for $\x_f$ and $E$ for two cases: when all nodes are controlled \xin{(i.e., $B=I$)} and when only a single node $i$ is controlled \xin{(i.e., $B=\e_i$)}, respectively.

\subsubsection{Control all nodes}
When all nodes are controlled, we have $B=I$. Thus, the optimal terminal point is
\begin{align*}
    \x_{f_{[n]}}=\frac{c}{nt_f}t_f\allone=\frac{c}{n}\allone,
\end{align*}
and the minimum energy is $E_{[n]}=\frac{c^2}{nt_f}.$

\subsubsection{Control a single node $i$}
When only a single node ${i\in[n]}$ is controlled, we have $B=\e_i$. Thus, the optimal terminal point is
\begin{align}
\label{eqn:xf_i}
    \x_{fi}=\frac{c}{t_f}\int_{0}^{t_f} e^{-L\tau}\e_i \,d\tau,
\end{align}
and the minimum energy is $E_i=\frac{c^2}{t_f}$.
\beh{This is rather a very interesting and surprising quantity: it shows that, when controlling a single node, the minimum energy $E_i$ required to land agents' states through such a Laplacian dynamics to the subspace $U=\{\allone^\intercal \x\geq c\}$ does not depend on the underlying graph, and the node in the network!}

Observe that when all nodes are controlled, the minimum energy control input steers the state to the consensus state $\frac{c}{n}\allone$. When only node $i$ is controlled, although the minimum energy $E_i$ remains constant for all $i$, the terminal points $\x_{fi}$ differ for each $i$. Therefore, \beh{it is natural to  consider a node that leads to a better homogeneity, a more \textit{central} node. More precisely, we can consider the distance $\|\x_{fi}-\frac{c}{n}\allone\|_2$ as a centrality measure.} In other words, given an opinion dynamics model over the network, considering the minimum energy control to drive the state from $\allzero$ to achieve an average state $\frac{c}{n}$ within time $t_f$, we rank the agents on the network by the distance between the respective terminal states and $\frac{c}{n}\allone$ (that is the homogeneous state derived by controlling all agents). \beh{Essentially, this centrality measure characterizes the  node's ability to unify the agents state.}

\begin{definition}
\label{def:Vc}
    We define \beh{the \ucenter} measure of a network $\G$ 
    to be the vector $V_c=\left(\|\x_{f1}-\frac{c}{n}\allone\|_2,\ldots, \|\x_{fn}-\frac{c}{n}\allone\|_2\right)$, where $\x_{fi}$ is the terminal point when only node $i$ is controlled \beh{i.e., $B=e_i$,} as defined in \eqref{eqn:xf_i}. We say that node $i$ is the $U$-central node if ${i\in \argmin_{j\in[n]}\|\x_{fj}-\frac{c}{n}\allone\|_2}$.
\end{definition}

Clearly, the \ucenter measure we define above depends on (i) the topology of the graph given by the adjacency matrix $\A$ or the Laplacian matrix $L$, and (ii) the termination time $t_f$. \beh{Interestingly, we show that for a fixed network, for the extreme cases where $t_f\approx0$ and $t_f\gg0$, \ucenter is closely related to existing centrality measures in network science}.


\section{The Main Results}
In this section, we present our two main results \xin{about how \ucenter relates to intrinsic centrality measures in two extreme control time horizons} and provide their proofs.
\begin{theorem}
   {For a connected undirected graph} and for short influence-time $t_f\approx0$, the \ucenter measure, as defined in Definition \ref{def:Vc}, coincides with the degree centrality. In other words, for sufficiently small $t_f$, a central node with respect to \ucenter is a node with the highest degree.
\end{theorem}
\begin{proof}
    Since $0$ is an eigenvalue of $L$ and its algebraic (and hence, geometric)  multiplicity is $1$ for a connected graph~\cite{von2007tutorial} and $L$ is positive semidefinite, we can diagonalize $L$ as
\begin{align*}
    L=\begin{bmatrix}
        \bu_1 & \cdots & \bu_n
    \end{bmatrix}
    \begin{bmatrix}
        0 & & & \\ & \lambda_2 & & \\ & & \ddots & \\ & & & \lambda_n
    \end{bmatrix}
    \begin{bmatrix}
        \bu_1^\intercal \\ \vdots \\ \bu_n^\intercal
    \end{bmatrix},
\end{align*}
where $0<\lambda_2\leq\cdots\leq\lambda_n$ are eigenvalues of $L$, and ${\bu_1=\frac{1}{\sqrt{n}}\allone},\ldots,\bu_n$ are its orthonormal eigenvectors. Therefore,
\begin{align*}
    e^{-L\tau}&=\begin{bmatrix}
        \bu_1 & \cdots & \bu_n
    \end{bmatrix}
    \begin{bmatrix}
        1 & & & \\ & e^{-\lambda_2\tau} & & \\ & & \ddots & \\ & & & e^{-\lambda_n\tau}
    \end{bmatrix}
    \begin{bmatrix}
        \bu_1^\intercal \\ \vdots \\ \bu_n^\intercal
    \end{bmatrix}\\
    &=\frac{1}{n}\allone\allone^\intercal+e^{-\lambda_2\tau}\bu_2\bu_2^\intercal+\cdots+e^{-\lambda_n\tau}\bu_n\bu_n^\intercal.
\end{align*}

Using \eqref{eqn:xf_i}, we have for any $t_f>0$,
\begin{align}
\label{eqn:dist_xfi_con}
    &\x_{fi}-\frac{c}{n}\allone
    =c\left(\frac{1}{t_f}\int_{0}^{t_f} e^{-L\tau}\e_i \,d\tau-\frac{1}{n}\allone\allone^\intercal\e_i\right)\cr
    =&\frac{c}{t_f}\int_0^{t_f}\left(e^{-L\tau}-\frac{1}{n}\allone\allone^\intercal\right) \,d\tau \e_i\cr
    =&\frac{c}{t_f}\int_0^{t_f}\left(e^{-\lambda_2\tau}\bu_2\bu_2^\intercal+\cdots+e^{-\lambda_n\tau}\bu_n\bu_n^\intercal\right) \,d\tau \e_i\cr
    =&\frac{c}{t_f}\left(\frac{1-e^{-\lambda_2t_f}}{\lambda_2}\bu_2\bu_2^\intercal+\cdots+\frac{1-e^{-\lambda_nt_f}}{\lambda_n}\bu_n\bu_n^\intercal\right)\e_i.\nonumber\\
\end{align}

\xin{
Note that \eqref{eqn:dist_xfi_con} is a function of $t_f$.
If we take the limit $t_f\rightarrow0$, by L'Hôpital's rule, we have for $j\in\{2,\ldots,n\}$,
\begin{align*}
    \lim_{t_f\rightarrow0}\frac{1-e^{-\lambda_jt_f}}{\lambda_j t_f}
    =\lim_{t_f\rightarrow0}\frac{\lambda_j e^{-\lambda_jt_f}}{\lambda_j}
    =1.
\end{align*}
Thus, the limit of \eqref{eqn:dist_xfi_con} is
\begin{align*}
    \lim_{t_f\rightarrow0}\x_{fi}-\frac{c}{n}\allone&=c\left(\bu_2\bu_2^\intercal+\cdots+\bu_n\bu_n^\intercal\right)\e_i\\
    =&c\left(I-\frac{1}{n}\allone\allone^\intercal\right)\e_i
    =c\left(\e_i-\frac{1}{n}\allone\right).
\end{align*}
\beh{Since the norm of such a  limit is independent of agent index, we need to investigate higher order terms. Therefore, consider the difference between~\eqref{eqn:dist_xfi_con} and its limit}
\begin{align*}
    &\x_{fi}-\frac{c}{n}\allone-c\left(\e_i-\frac{1}{n}\allone\right)\\
    =&c\!\left(\!\frac{1\!-\!e^{-\lambda_2t_f}\!-\!\lambda_2 t_f}{\lambda_2 t_f}\bu_2\bu_2^\intercal\!+\!\cdots\!+\!\frac{1\!-\!e^{-\lambda_nt_f}\!-\!\lambda_n t_f}{\lambda_n t_f}\bu_n\bu_n^\intercal\!\right)\!\e_i.
\end{align*}
Again, using L'Hôpital's rule, we have for $j\in\{2,\ldots,n\}$,
\begin{align*}
    &\lim_{t_f\rightarrow0}\frac{1-e^{-\lambda_jt_f}-\lambda_j t_f}{\lambda_j t_f^2}
    =\lim_{t_f\rightarrow0}\frac{\lambda_j e^{-\lambda_jt_f}-\lambda_j}{2\lambda_j t_f}\\
    =&\lim_{t_f\rightarrow0}\frac{-\lambda_j^2 e^{-\lambda_jt_f}}{2\lambda_j}
    =-\frac{\lambda_j}{2}.
\end{align*}
Thus, we have
\begin{align*}
    &\lim_{t_f\rightarrow0}\frac{1}{t_f}\left(\x_{fi}-\frac{c}{n}\allone-c\left(\e_i-\frac{1}{n}\allone\right)\right)\\
    =&-c\left(\frac{\lambda_2}{2}\bu_2\bu_2^\intercal+\cdots+\frac{\lambda_n}{2}\bu_n\bu_n^\intercal\right)\e_i
    =-\frac{c}{2}L\e_i.
\end{align*}
By continuity of norm function, we have
\begin{align*}
    \lim_{t_f\rightarrow0}\left\|\x_{fi}-\frac{c}{n}\allone\right\|_2
    =c\left\|\e_i-\frac{1}{n}\allone\right\|_2
    =c\sqrt{\frac{n-1}{n}},
\end{align*}
and thus
\begin{align}
\label{eqn:lim_dist_tf_small}
    &\lim_{t_f\rightarrow0}\frac{1}{t_f}\left(\left\|\x_{fi}-\frac{c}{n}\allone\right\|_2-c\sqrt{\frac{n-1}{n}}\right)\cr
    =&\lim_{t_f\rightarrow0}\left\|\frac{c}{t_f}\left(\e_i-\frac{1}{n}\allone\right)-\frac{c}{2}L\e_i\right\|_2-\frac{c}{t_f}\sqrt{\frac{n-1}{n}}\cr
    =&\lim_{t_f\rightarrow0}\frac{c}{t_f}\!\left(\!\frac{\sqrt{\!\left(\!n\!-\!\frac{1}{2}nt_fd_i\!-\!1\!\right)^2\!+d_i\!\left(\!\frac{1}{2}nt_f\!-\!1\!\right)^2\!+\!n\!-\!1\!-\!d_i}}{n}\right.\cr
    &\left.-\sqrt{\frac{n-1}{n}}\right)\cr
    \stackrel{\rm{(a)}}{=}&\lim_{t_f\rightarrow0}\frac{c}{2n}\frac{-n^2d_i+\frac{1}{2}n^2d_i(d_i+1)t_f}{\sqrt{\!\left(\!n\!-\!\frac{1}{2}nt_fd_i\!-\!1\!\right)^2\!+d_i\!\left(\!\frac{1}{2}nt_f\!-\!1\!\right)^2\!+\!n\!-\!1\!-\!d_i}}\cr
    =&-\frac{c}{2}\sqrt{\frac{n}{n-1}}d_i,
\end{align}
where L'Hôpital's rule is used again in (a).
From \eqref{eqn:lim_dist_tf_small} we know that when $t_f$ increases slightly from $0$, the node with higher degree $d_i$ has distance $\left\|\x_{fi}-\frac{c}{n}\allone\right\|_2$ that decreases faster from $c\sqrt{\frac{n-1}{n}}$. Thus, when $t_f$ is sufficiently small, the node with higher degree $d_i$ has smaller distance $\left\|\x_{fi}-\frac{c}{n}\allone\right\|_2$, i.e., the \ucenter measure coincides with degree centrality.
}

\end{proof}

\begin{remark}
    When $t_f\rightarrow0$, ${\x_{fi}\rightarrow c\e_i}$, and when $t_f\approx0$, ${\left\|\x_{fi}-\frac{c}{n}\allone\right\|_2}$ varies only slightly with $d_i$. This implies that for short-term influence, if node $i$ is controlled, the minimum energy control steers node $i$'s state $x_i$ from $0$ to approximately $c$, and the states of other nodes change slightly, approximately remaining $0$. This makes sense because the agents do not have enough time to interact. The nodes with higher degree $d_i$ can interact with more neighbors in short time, so $\x_{fi}$ can be slightly closer to consensus.
\end{remark}

Next, we are going to interpret \ucenter for the long-term influence. Before presenting our second result, we first introduce a centrality measure, not previously found in the literature, which relies only on the structure of the graph.

\begin{definition}(Laplacian Inverse Centrality)
\label{def:L_inv_centrality}
    Given an undirected graph $\G$ with Laplacian matrix $L$, let $L^\dagger$ be the Moore–Penrose inverse \cite{penrose1955generalized} of $L$. Then we define the Laplacian inverse centrality of node $i$ to be $\|L^\dagger\e_i\|_2$, i.e., the Euclidean norm of the $i$th column of $L^\dagger$. We say that node $i$ is the central node with respect to Laplacian inverse centrality if ${i\in \argmin_{j\in[n]}\|L^\dagger\e_j\|_2}$.
\end{definition}

As we will discuss later, $L^\dagger$ has been well studied in terms of resistance distance and current-flow closeness centrality \cite{gutman2004generalized,bozzo2013resistance}. Unlike previous work, we will show another interpretation of $L^\dagger$ for tree graphs.

\begin{theorem}
    For a connected undirected graph and for long-term influence $t_f\gg0$, the \ucenter measure as defined in Definition \ref{def:Vc} coincides with the Laplacian inverse centrality defined in Definition \ref{def:L_inv_centrality}. Consequently, for sufficiently large $t_f$, the central node with respect to \ucenter is the central node with respect to Laplacian inverse centrality.
\end{theorem}

\begin{proof}
    If we take the limit $t_f\rightarrow\infty$ in \eqref{eqn:dist_xfi_con}, we have
    \begin{align*}
        \x_{fi}-\frac{c}{n}\allone
        =\frac{c}{t_f}\!\left(\!\frac{1}{\lambda_2}\bu_2\bu_2^\intercal+\cdots+\frac{1}{\lambda_n}\bu_n\bu_n^\intercal\!\right)\!\e_i+\s_i,
    \end{align*}
    where $\s_i=\begin{bmatrix} s_{i1}&\cdots&s_{in} \end{bmatrix}^\intercal$ with ${s_{ij}=\Oc\!\left(\!\frac{1}{t_f e^{\lambda_2t_f}}\!\right)\!}$ for ${j\in [n]}$.
    
    Let $L^\ddagger=\frac{1}{\lambda_2}\bu_2\bu_2^\intercal+\cdots+\frac{1}{\lambda_n}\bu_n\bu_n^\intercal$. Note that $L^\ddagger$ is symmetric and positive semidefinite, and $L^\ddagger\allone=L^{\ddagger\intercal} \allone=\allzero$. In addition, we have 
    \begin{align*}
        LL^\ddagger=L^\ddagger L
        =\bu_2\bu_2^\intercal+\cdots+\bu_n\bu_n^\intercal
        \!=\!I-\bu_1\bu_1^\intercal
        \!=\!\!I-\!\frac{1}{n}\allone\allone^\intercal\!\!.
    \end{align*}
    Therefore, we can verify that $L^\ddagger$ satisfies all the Moore–Penrose conditions: $LL^\ddagger L=L$, $L^\ddagger LL^\ddagger=L^\ddagger$, $(LL^\ddagger)^\intercal=LL^\ddagger$, and $(L^\ddagger L)^\intercal=L^\ddagger L$. Thus, $L^\ddagger=L^\dagger$, i.e., $L^\ddagger$ is the Moore–Penrose inverse of $L$, and when $t_f\rightarrow\infty$,
    \begin{align*}
        \left\|\x_{fi}-\frac{c}{n}\allone\right\|_2
        =\frac{c}{t_f}\|L^\dagger\e_i\|_2+\Oc\left(\frac{1}{t_f e^{\lambda_2t_f}}\right).
    \end{align*}
    As a result, when $t_f$ is sufficiently large, \ucenter coincides with Laplacian inverse centrality.
\end{proof}

\begin{remark}
\label{rem:Linv}
    First, it should be noted that $\left\|\x_{fi}-\frac{c}{n}\allone\right\|_2$ decreases at the rate of $\Oc\left(\frac{1}{t_f}\right)$ as $t_f$ increases and when $t_f\rightarrow\infty$, $\x_{fi}\rightarrow\frac{c}{n}\allone$. This makes sense as the larger $t_f$ is, the more time agents have to interact through the network so that $\x_{fi}$ can be closer to the consensus $\frac{c}{n}\allone$.
    Second, it is more interesting to connect $\|L^\dagger\e_i\|_2$ with node $i$'s position in the graph. In the existing literature, $L^\dagger$ is particularly useful to define the resistance distance between nodes in a network \cite{bozzo2013resistance}. Consider a network in which the edges represent resistors and the nodes are junctions between resistors. Each edge has conductance $1$. Given nodes $i$ and $j$, the resistance distance $R_{ij}$ between nodes $i$ and $j$ is defined as the effective resistance between nodes $i$ and $j$, which is equal to the potential difference between nodes $i$ and $j$ when a current source of 1~Amp is placed between $i$ (input) and $j$ (output). In this case, it has been derived in \cite{bozzo2013resistance} that
    \begin{align}
    \label{eqn:resist_dist}
        R_{ij}=L^\dagger_{ii}+L^\dagger_{jj}-2L^\dagger_{ij}.
    \end{align}
    Therefore, the mean resistance of node $i$ from the other nodes would be
    \begin{align}
    \label{eqn:mean_resist_dist}
        \frac{1}{n}\sum_{j=1}^{n}R_{ij}=L^\dagger_{ii}+\frac{\trace(L^\dagger)}{n}.
    \end{align}
    Therefore, the diagonal element $L^\dagger_{ii}$, which is always nonnegative as $L^\dagger$ is positive semidefinite, is large if the mean resistance distance of node $i$ from the rest of the graph is high. From \eqref{eqn:resist_dist} we know the off-diagonal element $L^\dagger_{ij}$ also tells us something about the resistance proximity between nodes $i$ and $j$. $L^\dagger_{ij}$ is high (in particular positive) if the resistance distance between nodes $i$ and $j$ is low, and $L^\dagger_{ij}$ is low (in particular negative) if the resistance distance between nodes $i$ and $j$ is high. The above reasoning gives us a sense that $\|L^\dagger\e_i\|_2$ provides a measure of node $i$'s peripherality.
\end{remark}

Next, we derive a characterization of $\|L^\dagger\e_i\|_2$ for tree graphs in terms of the pairwise node distances on the graph.

    \subsection{\ucenter of Trees for $t_f\gg 0$}
    Recall that a tree is a  connected and cycle-free graph. Given any nodes $i$ and $j$ in a tree graph, there is only one path connecting them, and let us denote the length of this path (i.e., the number of edges on this path) by $d_{ij}$. Thus, for the resistor network we described earlier, the effective resistance between nodes $i$ and $j$ is $d_{ij}$. Let ${D_i:=\sum_{j=1}^n d_{ij}}$, and ${W:=\sum_{i=1}^n D_i}$. Replacing $R_{ij}$ by $d_{ij}$ in \eqref{eqn:resist_dist} and \eqref{eqn:mean_resist_dist}, we have
    \begin{align*}
        d_{ij}&=L^\dagger_{ii}+L^\dagger_{jj}-2L^\dagger_{ij},\\
        \frac{1}{n}D_i&=L^\dagger_{ii}+\frac{1}{n}\sum_{j=1}^n L^\dagger_{jj}.
    \end{align*}
    Solving these two equations, we have
    \begin{align*}
        L^\dagger_{ij}=\frac{1}{2}\left(\frac{D_i+D_j}{n}-d_{ij}-\frac{W}{n^2}\right)
    \end{align*}
    for all $i,j\in[n]$. Thus,
    \begin{align}
    \label{eqn:L_inv_vecnorm}
        &\|L^\dagger\e_i\|_2^2
        =\frac{1}{4}\sum_{j=1}^n\left(\frac{D_i+D_j}{n}-d_{ij}-\frac{W}{n^2}\right)^2\cr
        =&\frac{1}{4}\left[\frac{\sum_{j=1}^n D_j^2}{n^2}
        +\sum_{j=1}^n\left(\frac{D_i}{n}-d_{ij}\right)^2\right.\cr
        &\left.+2\sum_{j=1}^n \frac{D_j}{n}\left(\frac{D_i}{n}-d_{ij}\right)-\frac{W^2}{n^3}\right].
    \end{align}
Notice that the first and the last term in the right-hand side of \eqref{eqn:L_inv_vecnorm} are common for all $i\in [n]$. The second term in \eqref{eqn:L_inv_vecnorm} is the empirical ``variance'' of the pairwise distance between the node $i$ and the other nodes in the graph, i.e., the empirical variance of $d_{i1},\ldots,d_{in}$, and is small if node $i$ has distances to other nodes that are relatively uniform. Interpreting the third term in \eqref{eqn:L_inv_vecnorm} is more complicated. Notice that $\frac{D_i}{n}-d_{ij}$ is large (in particular positive) if node $j$ is close to $i$, and is small (in particular negative) if node $j$ is far from $i$. Thus, if node $i$ is central, for node $j$ that is close to $i$, node $j$ is also central, so $D_j$ is small and $\frac{D_i}{n}-d_{ij}$ is large; for node $j$ that is far from $i$, node $j$ is peripheral, so $D_j$ is large and $\frac{D_i}{n}-d_{ij}$ is small. Therefore, a more central node has a smaller third term. In sum, if node $i$ is more central and other nodes are more "evenly" positioned relative to it, then $\|L^\dagger\e_i\|_2$ is smaller.

\section{Numerical Experiments}

In this section, we present three examples to illustrate our results: a tree network, the Minnesota road network, and a Facebook network (https://snap.stanford.edu/data/egonets-Facebook.html). In these examples, we set the sum of terminal states to $c=1$, and compute $\left\|\x_{fi}-\frac{c}{n}\allone\right\|_2$ for each node $i$ for the short-term influence $t_f=0.01$ or $0.001$ and the long-term influence $t_f=1000$, respectively. We compare our \ucenter measure as defined in Definition \ref{def:Vc} with the following 6 centrality measures.
\begin{itemize}
    \item Degree centrality: The centrality of node $i$ is simply its degree (i.e., the number of nodes connected to $i$). The central nodes are the nodes with the highest degree.
    \item Eigenvector centrality: Let $\A$ be the adjacency matrix for a strongly connected graph. Then, by the Perron-Frobenius Theorem (see Chapter 7 in \cite{meyer2023matrix}), $\A$ has a unique eigenvector $\bv$ with positive entries (up to multiplication). The eigenvector centrality of node $i$ is the $i$th component $v_i$ of $\bv$. The central nodes are the nodes with the highest $v_i$.
    \item Closeness centrality: Given nodes $i$ and $j$, let the length of the shortest path connecting them be $d_{ij}$. The closeness centrality of node $i$ is ${D_i:=\sum_{j=1}^n d_{ij}}$. The central nodes are the nodes with the lowest $D_i$.
    \item Variance centrality: The variance centrality of node $i$ is $\sum_{j=1}^n\left(\frac{D_i}{n}-d_{ij}\right)^2$. The central nodes are the nodes with the lowest variance. This definition of centrality has not been introduced in existing literature. We are inspired to study it because it appears as a term in \eqref{eqn:L_inv_vecnorm}.
    \item Current-flow closeness centrality: Recall the definition of resistance distance $R_{ij}$ between nodes $i$ and $j$ in Remark \ref{rem:Linv}. The current-flow closeness centrality of node $i$ is $R_i=\sum_{j=1}^{n}R_{ij}$. The central nodes are the nodes with the lowest $R_i$. For tree graphs, current-flow closeness centrality is the same as closeness centrality.
    \item Current-flow variance centrality: The current-flow variance centrality of node $i$ is $\sum_{j=1}^n\left(\frac{R_i}{n}-R_{ij}\right)^2$. The central nodes are the nodes with the lowest current-flow variance. For tree graphs, current-flow variance centrality is the same as variance centrality.
\end{itemize}

\begin{figure}[h]
    \centering
    \includegraphics[width=\linewidth]{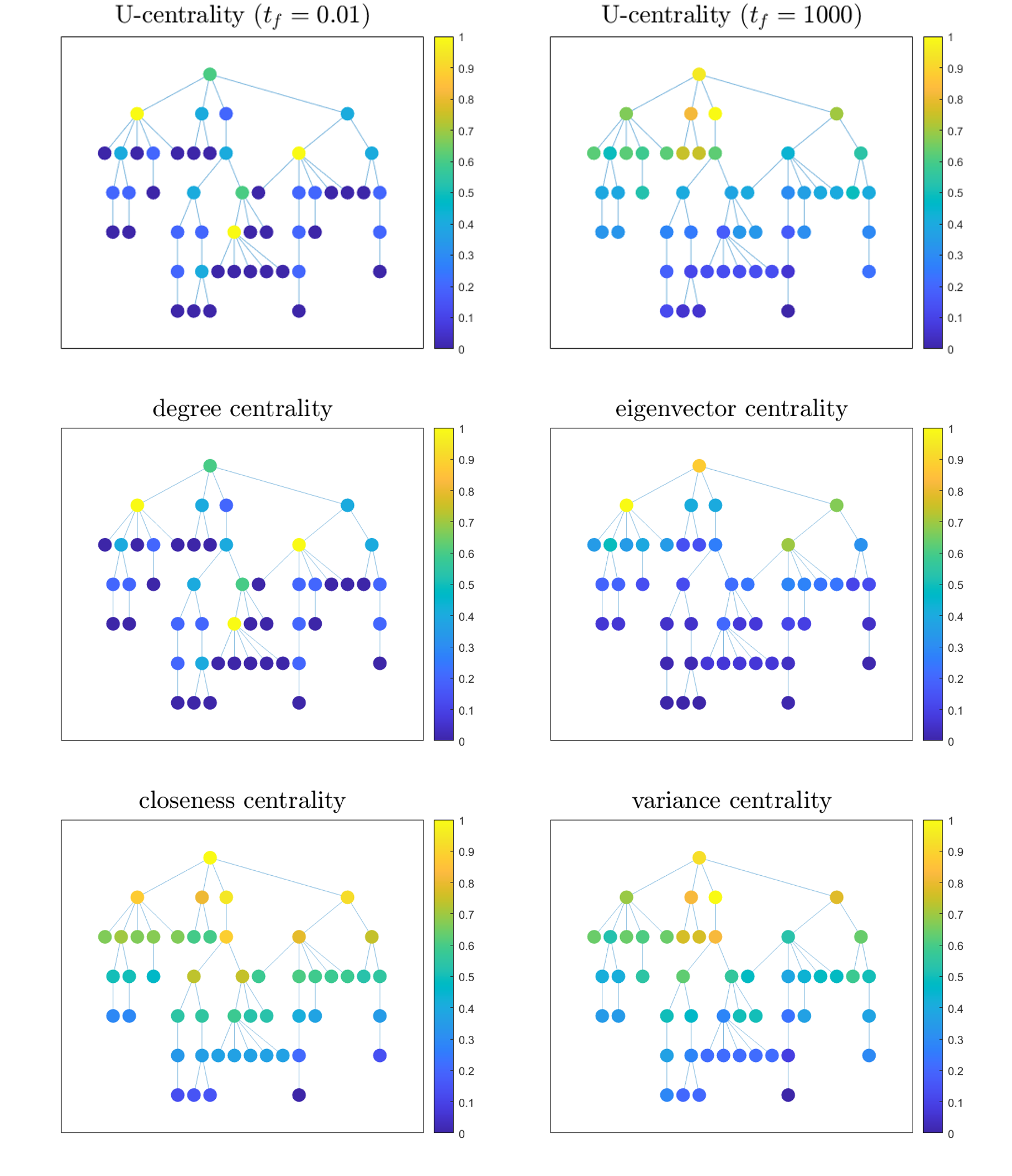}
    \caption{Comparison between \ucenter and other centrality measures for a tree graph.}
    \label{fig:simu_tree}
\end{figure}

\begin{figure}[h]
    \centering
    \includegraphics[width=\linewidth]{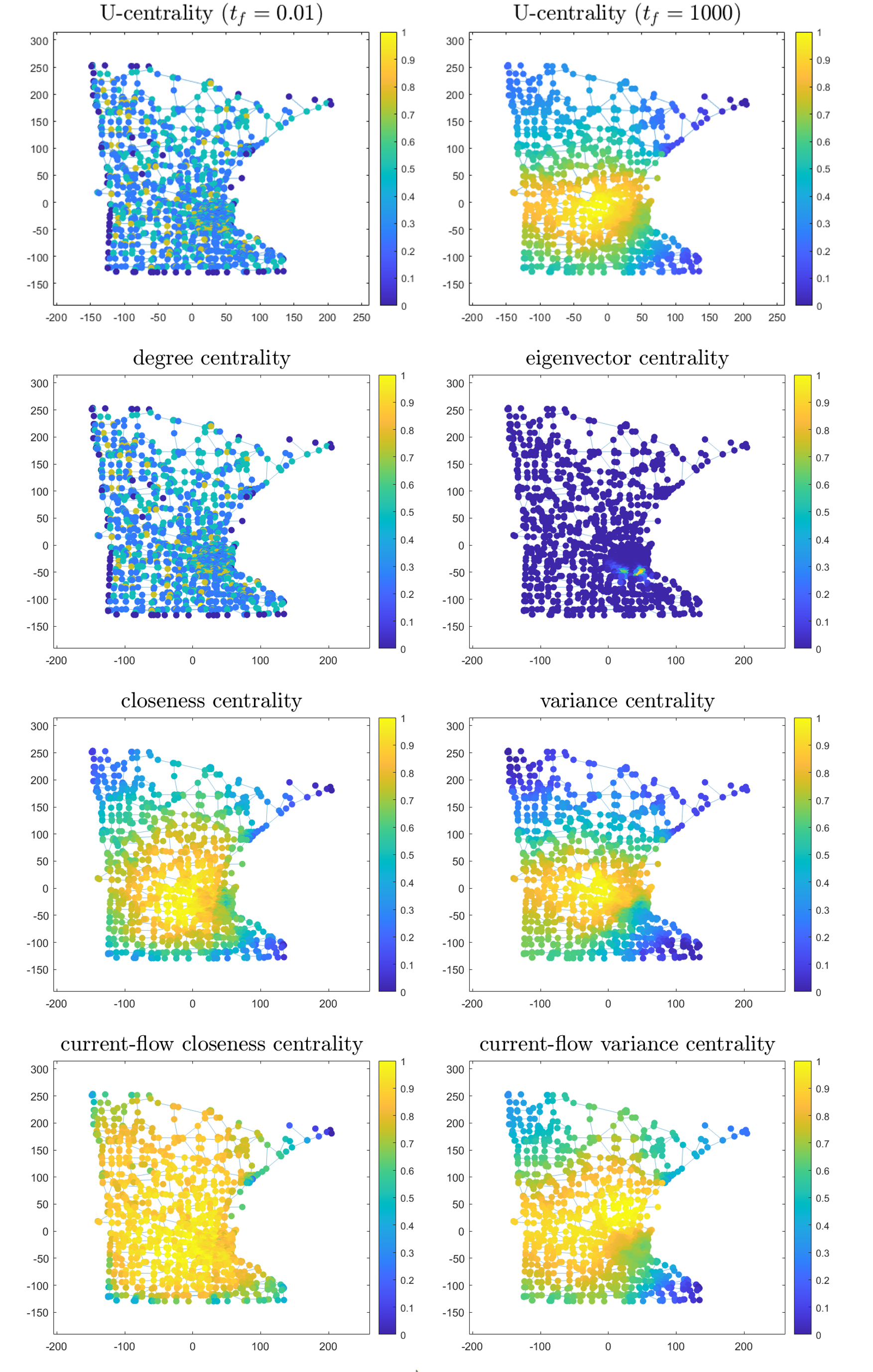}
    \caption{Comparison between \ucenter and other centrality measures for the Minnesota road network.}
    \label{fig:simu_Minnesota}
\end{figure}

\begin{figure}[h]
    \centering
    \includegraphics[width=\linewidth]{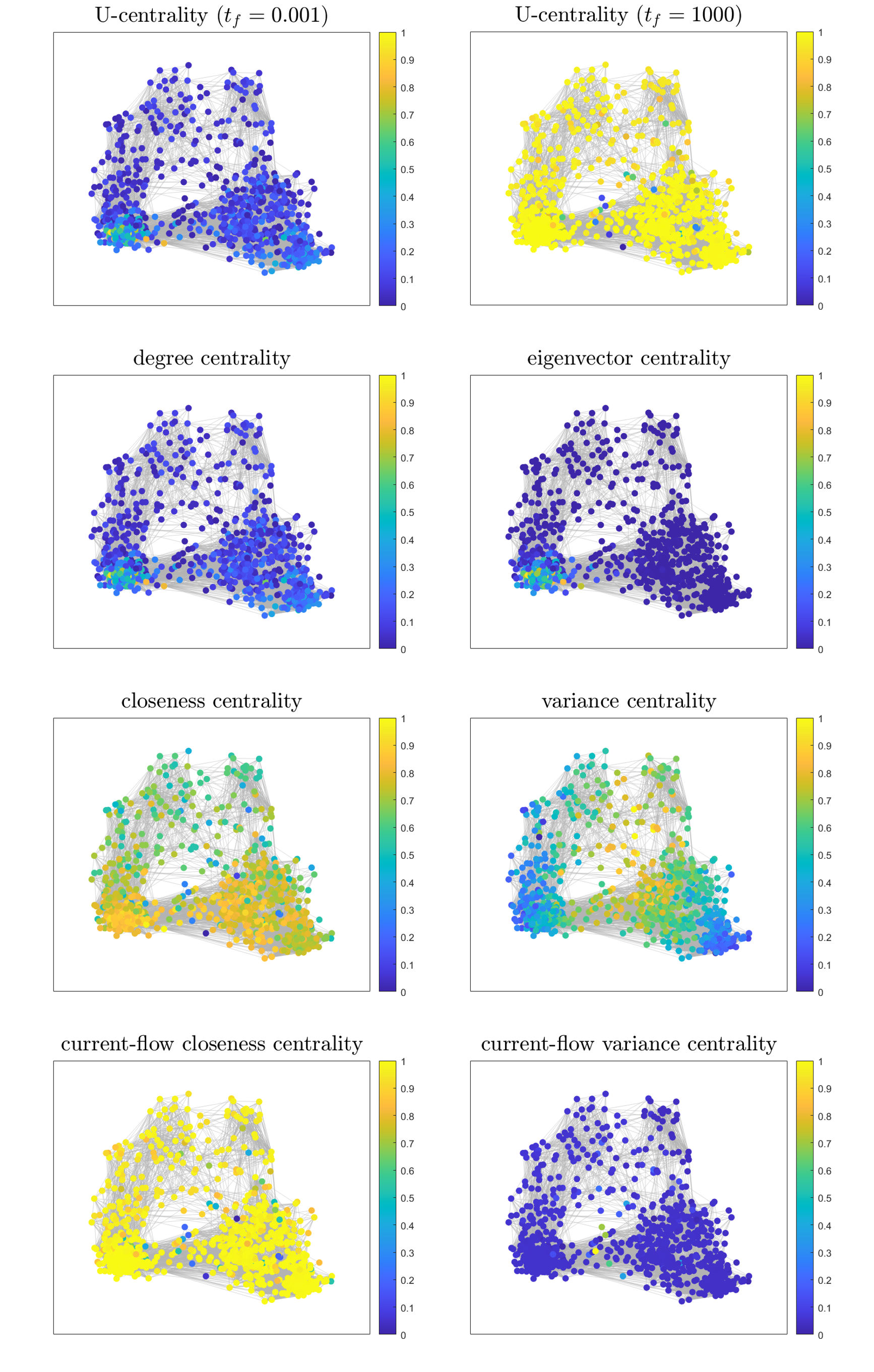}
    \caption{Comparison between \ucenter and other centrality measures for a Facebook network.}
    \label{fig:simu_fb}
\end{figure}

The simulation results for a tree graph with $n=50$ nodes, for the Minnesota road network, and for the largest connected component of a Facebook network are shown in Figure \ref{fig:simu_tree}, \ref{fig:simu_Minnesota}, and \ref{fig:simu_fb} respectively, where warmer colors represent more central nodes, and cooler colors represent less central nodes. In both examples, we observe that when $t_f$ is small, \ucenter coincides with degree centrality. When $t_f$ is large, \ucenter is distinct from all other centrality measures but bears the closest resemblance to variance centrality or current-flow closeness centrality. Generally speaking, when $t_f$ is large, \ucenter reflects the position of the nodes in relation to their peripherality, confirming our interpretation outlined earlier.

To examine how \ucenter behaves across different time scales, we provide a video (https://youtu.be/YY2LluJScts) illustrating the evolution of \ucenter as $t_f$ increases.


\section{Conclusion}

We proposed to view network centrality in the context of optimal control of dynamical system over the underlying network describing the context in which such a centrality notion is used. As a result, the centrality would depend on the \textit{dynamics}, and the \textit{control objective}. To showcase the strength of this framework, we proposed and studied a centrality measure based on minimum energy control (from different nodes) over Laplacian dynamics. We showed that the resulting centrality measure, which depends on the time-scale under investigation, would be closely related to existing centrality measures in network science in two extreme time horizons.

A natural next step in our research is to conduct a theoretical investigation of how \ucenter evolves over time.
Some other potential directions for future investigations include studying non-Laplacian dynamics over networks and adopting different control objectives such as minimum time control. It is also interesting to investigate that if a subset of agents can be controlled, which agents are the best agents to control.



\bibliographystyle{abbrv}
\bibliography{bib}

\end{document}